\documentclass[12pt]{iopart}
\usepackage{amsthm,graphicx}
\usepackage{amssymb}
\def\softd{{\leavevmode\setbox1=\hbox{d}%
\hbox to 1.05\wd1{d\kern-0.4ex{\char039}\hss}}}%cstocs
\def\softt{{\leavevmode\setbox1=\hbox{t}%
\hbox to \wd1{t\kern-0.6ex{\char039}\hss}}}%cstocs

\newcommand{\R}{\mathbb{R}}
\newcommand{\D}{\mathrm{d}}

\newtheorem{theorem}{Theorem}[section]
\newtheorem{proposition}{Proposition}[section]
\newtheorem{remark}{Remark}[section]

\renewcommand{\theequation}{\arabic{section}.\arabic{equation}}
\begin{document}

\title[Spectral estimates for a class of Schr\"odinger operators]
{Spectral estimates for a class of Schr\"odinger operators with infinite phase space and potential unbounded from below}

%    Information for first author
\author{Pavel Exner and Diana Barseghyan}
\address{Doppler Institute for Mathematical Physics and Applied
Mathematics, \\ B\v{r}ehov\'{a} 7, 11519 Prague, \\ and  Nuclear
Physics Institute ASCR, 25068 \v{R}e\v{z} near Prague, Czechia}
\ead{exner@ujf.cas.cz, dianabar@bk.ru}

\begin{abstract}
We analyze two-dimensional Schr\"odinger operators with the potential $|xy|^p - \lambda (x^2+y^2)^{p/(p+2)}$ where $p\ge 1$ and $\lambda\ge 0$. We show that there is a critical value of $\lambda$ such that the spectrum for $\lambda<\lambda_\mathrm{crit}$ is below bounded and purely discrete, while for $\lambda>\lambda_\mathrm{crit}$ it is unbounded from below. In the subcritical case we prove upper and lower bounds for the eigenvalue sums.
\end{abstract}

%Uncomment for PACS numbers title message
%\pacs{00.00, 20.00, 42.10}
% Keywords required only for MST, PB, PMB, PM, JOA, JOB?
%\vspace{2pc}
%\noindent{\it Keywords}: Article preparation, IOP journals
% Uncomment for Submitted to journal title message
%\submitto{\JPA}
% Comment out if separate title page not required
\maketitle

%%%%%%%%%%%%%%%%%
%%%% Introduction
%%%%%%%%%%%%%%%%%
\section{Introduction} \label{s: intro}

While the idea of Hermann Weyl to analyze spectra of quantum systems semiclassically, by looking at the phase space allowed for the corresponding classical motion, is one of the most seminal in modern mathematical physics, its validity is not universal. Various examples of systems which have purely discrete spectrum despite the fact the respective phase space volume is infinite were constructed in the last three decades. A classical one belongs to B.~Simon \cite{Si83} and describes a two-dimensional Schr\"odinger operator with the potential $|xy|^p$ having deep ``valleys'' the width of which is shrinking with the distance from the origin. A related problem concerns spectral properties of Dirichlet Laplacians in regions with hyperbolic cusps --- we refer to the recent paper \cite{GW11} for an up-to-date bibliographical survey.

The aim of the present paper is two-fold. First, we want to demonstrate that similar spectral behaviour can occur also for Schr\"odinger operators with potentials unbounded from below. Furthermore, we intend to construct a model which exhibit a nontrivial spectral transition as the coupling constant changes. Specifically, we are going to consider here the following class of operators,
 % ------------- %
 \begin{equation} \label{operator}
 L_p(\lambda)\,:\; L_p(\lambda)\psi= -\Delta\psi + \left( |xy|^p - \lambda (x^2+y^2)^{p/(p+2)} \right)\psi\,, \quad p\ge 1\,,
 \end{equation}
 % ------------- %
on $L^2(\R^2)$ using the standard Cartesian coordinate $(x,y)$ in $\R^2$; the parameter $\lambda$ controlling the second term of the potential is non-negative and we will not indicated it if it will be clear from the context. Since $\frac{2p}{p+2}<2$ the above operator is essentially self-adjoint on $C_0^\infty(\R^2)$ by Faris-Lavine theorem -- cf. \cite{RS75}, Thms.~X.28 and  X.38; in the following the symbol $L_p$ or $L_p(\lambda)$ will always mean its closure.

First we will show that there is a critical value of the coupling constant $\lambda$, expressed explicitly as ground-state eigenvalue of the corresponding (an)harmonic oscillator Hamiltonian, such that the spectrum of $L_p(\lambda)$ is below bounded and purely discrete for $\lambda< \lambda_\mathrm{crit}$, while for $\lambda>\lambda_\mathrm{crit}$ it becomes unbounded from below. In the latter case one naturally expects it to be continuous covering the whole real axis but we will not proceed this way. The main result of the paper are upper and lower bounds to the sums of the first $N$ eigenvalues of $L_p(\lambda)$ in the subcritical case proved in Sec.~\ref{s: estimates}.

They give the same asymptotics up to a multiplicative constant if $\lambda=0$, while for $\lambda>0$ this remains true for the leading term but an additional one, linear in $N$, is added in the lower bound. The proof of the upper bound is reduced to the case $\lambda=0$, hence it is not surprising the result holds for any $\lambda< \lambda_\mathrm{crit}$. On the other hand, the argument which yields the lower bound is more subtle and we have been able to prove the result for sufficiently small values of $\lambda$ only leaving room for improvement, see Theorem~\ref{th:lowerbound} below. In the closing section we prove also a lower bound to spectral sums of the operator $-\Delta_D-\lambda(x^2+y^2)$ for $0\le\lambda<1$, where $-\Delta_D$ is the Dirichlet Laplacian on the region with hyperbolic cusps; formally speaking this can be regarded as the limit $p\to\infty$ of the problem (\ref{operator}).

%%%%%%%%%%%%%%%%%%%%%%%%%%%%%%%
%%%% Discreteness of the spectrum
%%%%%%%%%%%%%%%%%%%%%%%%%%%%%%%
\section{Discreteness of the spectrum}\label{s: discreteness}
\setcounter{equation}{0}

The first important observation is that spectral properties of the operator $L_p(\lambda)$ depend crucially on the value of the parameter $\lambda$; we have to distinguish two cases.

\subsection{The subcritical case}
\label{ss: subcrit}

The spectral regime we are primarily interested in occurs for small values of $\lambda$. To characterize the smallness quantitatively we need an auxiliary operator which will be an (an)harmonic oscillator Hamiltonian on line,
 % ------------- %
 \begin{equation} \label{oscillator}
 \tilde H_p\,:\: \tilde H_p u = -u''+|t|^p u
 \end{equation}
 % ------------- %
on $L^2(\R)$ with the standard domain. Let $\gamma_p$ be the minimal eigenvalue of this operator; in view of the potential symmetry we have $\gamma_p = \inf \sigma(H_p)$, where
 % ------------- %
 \begin{equation} \label{halfoscillator}
 H_p\,:\: H_p u = -u''+t^p u
 \end{equation}
 % ------------- %
on $L^2(\R_+)$ with Neumann condition at $t=0$. This quantity is well known --- cf.~Fig.~\ref{gammap} --- for $p=2$ is equals one, it reaches its minimum $\gamma_p\approx 0.998995$ at $p\approx 1.788$ and grows to $\frac14 \pi^2$ as $p\to\infty$.
 % ------------- %
\begin{figure}[t]
\vspace{2em}
\begin{center}
\includegraphics[angle=0,width=0.5\textwidth]{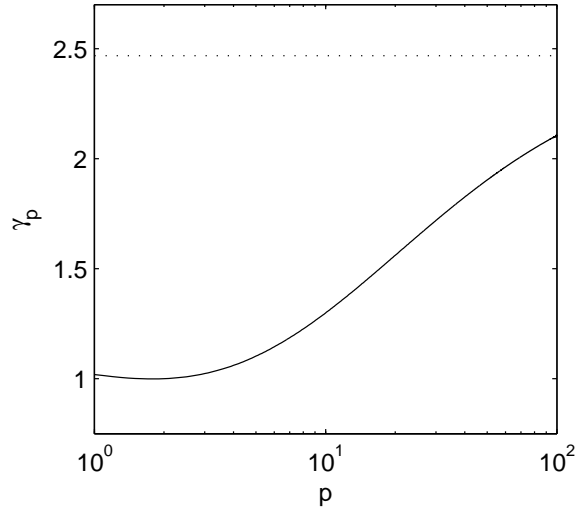} \vspace{0em}
\caption{$\gamma_p$ as a function of $p$ \label{gammap} in semilogarithmic scale.}
\end{center}
\end{figure}
 % ------------- %
We have the following result.

 % ------------- %
 \begin{theorem} \label{thm:subcrit}
 For any $\lambda\in [0,\lambda_\mathrm{crit})$, where $\lambda_\mathrm{crit} := \gamma_p$, the operator $L_p(\lambda)$ has a purely discrete spectrum bounded from below for any $p\ge 1$ .
 \end{theorem}
 % ------------- %
\begin{proof}[Proof:]
By the minimax principle it is sufficient to estimate $L_p$ from below by a self-adjoint operator with a purely discrete spectrum. To construct such a lower bound we employ bracketing imposing additional Neumann condition at concentric circles of radii $n=1,2,\dots\,$. Using the standard polar coordinates, we obtain thus a direct sum of operators acting as
 % ------------- %
 \begin{eqnarray}
\lefteqn{L^{(1)}_{n,p}\psi=-\frac{1}{r}\frac{\partial}{\partial r}\biggl(r\frac{\partial\psi}{\partial r}
\biggr)-\frac{1}{n^2}\frac{\partial^2\psi}{\partial\varphi^2}+\biggl(\frac{r^{2p}}{2^p}
|\sin2\varphi|^p-\lambda r^{2p/(p+2)}\biggr)\psi\,,} \\ [.5em] &&
\hspace{4em}\left. \frac{\partial\psi}{\partial n}\right|_{r=n-1}= \left.
\frac{\partial\psi}{\partial n}\right|_{r=n}=0 \,, \label{L_1}
 \end{eqnarray}
 % ------------- %
on the annular regions $G_n:=\{(r,\varphi)\,:\: n-1\leq r<n,\,\,0\leq\varphi<2\pi\},\, n=1,2,\ldots\:$. Each of the latter is compact and the potential is regular on it, hence $\sigma(L^{(1)}_{n,p})$ is purely discrete. It is sufficient therefore to check that $\inf\sigma(L^{(1)}_{n,p})\to \infty$ as $n\to\infty$, because then the spectrum of $\bigoplus_{n=1}^\infty L^{(1)}_{n,p}$ below any fixed value is a finite union of discrete spectra which implies the sought claim.

The argument can further simplified if we estimate $L^{(1)}_{n,p}$ from below by an operator with separating variables acting, for instance, as
 % ------------- %
 \begin{eqnarray*}
\lefteqn{L^{(2)}_{n,p}\psi =-\frac{1}{r}\frac{\partial}{\partial r}\biggl(r\frac{\partial\psi}{\partial r}
\biggr)-\frac{1}{n^2}\frac{\partial^2\psi}{\partial\varphi^2}+\biggl(\frac{(n-1)^{2p}}{2^p}
|\sin2\varphi|^p-\lambda n^{2p/(p+2)}\biggr)\psi\,,} \\ [.5em] &&
\hspace{4em}\left. \frac{\partial\psi}{\partial n}\right|_{r=n-1}= \left.
\frac{\partial\psi}{\partial n}\right|_{r=n}=0 \,,
 \end{eqnarray*}
 % ------------- %
the spectrum of which is the ``sum'' of the radial and angular component. Furthermore, the lowest radial eigenfunction is zero corresponding to a constant eigenfunction, hence the problem reduces to analysis of the angular component of the above operator; using the mirror symmetries of the potential it is enough to consider in on $L^2(0,\pi/4)$ with Neumann conditions at the endpoints of the interval.

To simplify things further we take an arbitrary $\varepsilon \in(0,1)$ and introduce the following ordinary differential operator on $L^2(0,\pi/4)$,
 % ------------- %
 $$
L^{(3)}_{n,p}\,:\: L^{(3)}_{n,p}u=-u''+\biggl(\frac{n^{2p+2}}{2^p}\sin^p2x -\frac{\lambda}{1-\varepsilon}\, n^{(4p+4)/(p+2)}\biggr)u
 $$
 % ------------- %
with Neumann boundary conditions, $u'(0)=u'(\pi/4)=0\,$. It is clear that for $n$ large enough, specifically $n > \left( 1- (1-\varepsilon)^{(p+2)/ (4p+4)} \right)^{-1}$, we have $n^2 L^{(2)}_{n,p} \ge L^{(3)}_{n-1,p}$, hence it is sufficient to investigate the spectral threshold $\mu_{n,p}$ of $L^{(3)}_{n,p}$. We employ one more estimate from below. To a fixed $\varepsilon$ there is $\delta(\varepsilon)$ such that $\sin2x\ge2(1 -\varepsilon)x$ holds for $0\le x\le\delta(\varepsilon)$. We consider the operator
 % ------------- %
$$
L^{(4)}_{n,p}:=-\frac{\D^2}{\D x^2}+
n^{2p+2}x^p\left(\chi_{(0,\delta(\varepsilon)]}(x)+
\left(\frac{2}{\pi}
\right)^p\chi_{[\delta(\varepsilon),\pi/4)}(x)\right) -\lambda_\varepsilon'\,n^{(4p+4)/(p+2)}
$$
 % ------------- %
with Neumann boundary conditions, where $\lambda_\varepsilon':= \lambda(1-\varepsilon)^{-p-1}$, and denote its lowest inequality as $\mu_{n,p}'$. In view of the inequality $\sin x\ge\frac{2}{\pi}x$ we have $L^{(3)}_{n,p}\ge(1-\varepsilon)^pL^{(4)}_{n,p}$, and thus $\mu_{n,p}\ge(1-\varepsilon)^p\mu_{n,p}'\,,\,n=1,2,\ldots\,$, by the minimax principle. It is therefore sufficient to check that
 % ------------- %
 \begin{equation} \label{blowup_bound}
\frac{\mu_{n,p}'}{n^2}\rightarrow\infty\qquad\textrm{as}\qquad n\rightarrow\infty\,.
 \end{equation}
 % ------------- %
It is straightforward to check that the operator $L^{(4)}_{n,p}$ without the last term is by a simple scaling transformation, $x=t\cdot n^{-(2p+2)/(p+2)}$, unitarily equivalent to the $\nu_p^2$ multiple of the operator
 % ------------- %
 \begin{equation} \label{osc_bound}
H_{n,p}=-\frac{\D^2}{\D t^2}+t^p\left(\chi_{(0,\nu_p
\delta(\varepsilon)]}(t)+\left(\frac{2}{\pi}\right)^p\chi_{[\nu_p
\delta(\varepsilon),\nu_p\pi/4)}(t)\right)
 \end{equation}
 % ------------- %
on $L^2\left(0,\frac14 \nu_p\pi\right)$ with Neumann conditions at the endpoints of the interval, where $\nu_p:=n^{(2p+2)/(p+2)}$. Spectrum of the above operator is purely discrete, and denoting by $\gamma_{n,p}$ its lowest eigenvalue, we thus have
 % ------------- %
 $$
\mu'_{n,p}=(\gamma_{n,p}-\lambda'_\varepsilon)\,n^{(4p+4)/(p+2)}\,.
 $$
 % ------------- %
Furthermore, we have $H_{n,p}^\mathrm{N} \le H_{n,p} \le H_{n,p}^\mathrm{D}$ where the estimating operators act as (\ref{osc_bound}) with additional Neumann and Dirichlet condition, respectively, at the point $t=\nu_p\delta(\varepsilon)$. For large enough $n$ the ground state of both the estimating operators come from the ``inner'' part, $t\in(0,\nu_p\delta(\varepsilon))$, and one can check explicitly that $\gamma_{n,p}^j \to \gamma_p$ as $n\to\infty$ for both $j=\mathrm{N,D}$, where $\gamma_p$ is the ground state of (\ref{halfoscillator}), and therefore $\gamma_{n,p} \to \gamma_p$ as $n\to\infty$. This further implies that to any positive $\varepsilon'$ there is a natural number $n_{\varepsilon'}$ such that
 % ------------- %
 $$
(\gamma_p-\lambda'_\varepsilon-\varepsilon')\, n^{(4p+4)/(p+2)} \le \mu'_{n,p}
\le (\gamma_p-\lambda'_\varepsilon+\varepsilon')\, n^{(4p+4)/(p+2)}
 $$
 % ------------- %
holds for for $n>n_{\varepsilon'}$. Since $\varepsilon$ and $\varepsilon'$ are arbitrary we conclude that (\ref{blowup_bound}) holds for any $p\ge 1$ whenever $\lambda < \gamma_p$ which establishes the claim of the theorem. \end{proof}

\subsection{The supercritical case}
\label{ss: supercrit}

For large values of $\lambda$ the picture changes.

 % ------------- %
 \begin{proposition}
 The spectrum of $L_p(\lambda),\: p\ge 1\,$, is below unbounded if $\lambda > \lambda_\mathrm{crit}$.
 \end{proposition}
 % ------------- %
\begin{proof}[Proof:] We employ a similar technique, this time looking for an upper estimate to $L_p(\lambda)$. We construct it by Dirichlet bracketing considering the operators acting as (\ref{L_1}) on the annular domains $G_n$, denoted again as $L^{(1)}_{n,p}$, this time with Dirichlet boundary conditions, $\left. \psi\right|_{\partial G_n}=0$. We have $\bigoplus_{n=1}^\infty L^{(1)}_{n,p} \ge L_p(\lambda)$, hence by minimax principle we have to prove that
 % ------------- %
 $$
\inf\, \sigma\!\left( L^{(1)}_{n,p} \right) \to -\infty \qquad \mathrm{as} \qquad n\to\infty\,.
 $$
 % ------------- %
We estimate $L^{(1)}_{n,p}$ from above by an operator with separated variables acting as
 % ------------- %
 $$
L^{(2)}_{n,p}\psi=-\frac{1}{r}\frac{\partial}{\partial r}\biggl(r\frac{\partial\psi}{\partial r}
\biggr)-\frac{1}{(n-1)^2}\frac{\partial^2\psi}{\partial\varphi^2}+
\biggl(\frac{n^{2p}}{2^p}|\sin2\varphi|^p-\lambda\,(n-1)^{2p/(p+2)}\biggr)\psi
 $$
 % ------------- %
on $H^2$ functions satisfying $\psi(n-1,\varphi)= \psi(n,\varphi) =0$ for all $\varphi \in [0,2\pi)$. The contribution from the radial term is now nonzero, and moreover, it depends on $n$, however, it is uniformly bounded. Specifically, the spectral threshold of $-\frac{1}{r}\frac{\partial}{\partial r}r\frac{\partial}{\partial r}$ on $L^2(n-1,n)$ with Dirichlet condition does not exceed $\pi^2$, the bound being saturated as $n\to\infty\:$ \cite{EFK04}. Consequently, it is sufficient to check that the spectral threshold of the angular part, or of a suitable one-dimensional operator estimating it from above tends to $-\infty$ as $n\to\infty$.

The argument is similar to the one used in the previous proof. We fix $\varepsilon\in(0,1)$ and analyze the operator acting as
 % ------------- %
 $$
L^{(2)}_{n,p}u=-u''(\varphi)+\biggl(\frac{n^{2p+2}}{2^p}
\sin^p2\varphi-(1-\varepsilon)\lambda\,n^{(4p+4)/(p+2)}\biggr)u
 $$
 % ------------- %
on $L^2(0,\pi/4)$ with Neumann conditions at the endpoints of the interval. Using the inequality $\sin x\le x$ on $[0,\pi/4]$ and the unitary equivalence given by the same scaling transformation as in the previous case we reduce the problem to investigation of the operator $\nu_p^2 H_{n,p}$ where $H_{n,p}u=-u''+t^pu$ on $L^2\left(0,\frac14 \nu_p\pi\right)$ with Neumann conditions at the endpoints. Denoting $\mu_{n,p}:= \inf \sigma \left(L^{(2)}_{n,p} \right)$ we find in the same way as above that to any $\varepsilon'>0$ there is a natural $n_{\varepsilon'}$ such that
 % ------------- %
 $$
(\gamma_p-(1-\varepsilon)\lambda-\varepsilon')n^{(4p+4)/(p+2)}\le\mu_{n,p}
\le(\gamma_p-(1-\varepsilon)\lambda+\varepsilon')n^{(4p+4)/(p+2)}
 $$
 % ------------- %
holds for all $n>n_{\varepsilon'}$. Since $\lambda>\gamma_p$ by assumption and $\varepsilon,\, \varepsilon'$ are arbitrary, the second inequality yields the desired result.
\end{proof}

%%%%%%%%%%%%%%%%%%%%%%%%%%%%%%%
%%%% Spectral estimates
%%%%%%%%%%%%%%%%%%%%%%%%%%%%%%%
\section{Spectral estimates}\label{s: estimates}
\setcounter{equation}{0}

Now we can pass to our main subject which is estimating eigenvalue sums of the operator (\ref{operator}) for small values of the coupling constant $\lambda$.

\subsection{Lower bounds to eigenvalue sums}
\label{ss: lower}

To state our result on lower bound on the spectrum we introduce the following quantity,
 % ------------- %
 $$
\alpha:= \frac{1}{40}\left(5+\sqrt{105}\right)^2 \approx 5.81\:.
 $$
 % ------------- %
It is clear from Fig.~\ref{gammap} that $\alpha^{-1}<\gamma_p$.
We denote by $\{\lambda_{j,p}\}_{j=1}^\infty$ the eigenvalues of $L_p(\lambda)$ arranged in the ascending order; then we can make the following claim.

 % ------------- %
 \begin{theorem} \label{th:lowerbound}
 To any nonnegative $\lambda<\alpha^{-1}\approx 0.172$ there exist a positive constant $C_p$ depending on $p$ only such that the following estimate is valid,
 % ------------- %
 \begin{equation} \label{lowerbound}
\sum_{j=1}^N\lambda_{j,p}\geq C_p(1-\alpha\lambda)\frac{N^{(2p+1)/(p+1)}}
{(\ln^pN+1)^{1/(p+1)}}-c\lambda\,N,\quad N=1,2,\ldots, %\qquad(3.1)
 \end{equation}
 % ------------- %
where $c=2\left(\frac{\alpha^2}{5}+1\right)\approx 15.51$.
 \end{theorem}
 % ------------- %
\begin{proof}[Proof:]
We denote by $\{\psi_{j,p}\}_{j=1}^\infty$ the system of normalized eigenfunctions corresponding to $\{\lambda_{j,p} \}_{j=1}^\infty$, i.e. we have
 % ------------- %
$$
-\Delta\psi_{j,p}+(|xy|^p-\lambda(x^2+y^2)^{p)/(p+2)})\psi_{j,p}=
\lambda_{j,p}\psi_{j,p},\quad j=1,2,\ldots\:;
$$
 % ------------- %
without loss of generality we may assume that the functions $\psi_{j,p}$ are real-valued. Our potential form hyperbolic-shaped ``valleys'' and our first task will be to find estimates on eigenfunction integrals in some corresponding regions. Specifically, we are going to demonstrate that for any natural number $j$ and a positive $\delta$ one has
 % ------------- %
 \begin{eqnarray}
\lefteqn{\hspace{-6em}\int_1^\infty\int_0^{(1+\delta)y^{-p/(p+2)}} \hspace{-1em} y^{2p/(p+2)}\psi_{j,p}^2(x,y)\,\D x\,\D y \le
\frac{5}{2}(1+\delta)^2\int_1^\infty\int_0^\infty\left(\frac{\partial\psi_{j,p}}{\partial x}\right)^2\!\!(x,y)
\,\D x\,\D y } \label{1st_est} \\[.5em] && %\qquad(3.2)
+2\frac{1+\delta}{\delta}\int_1^\infty\int_0^{(1+\delta) y^{-p/(p+2)}}x^py^p\psi_{j,p}^2(x,y)\,\D x\,\D y \nonumber
 \end{eqnarray}
 % ------------- %
and that for an arbitrary $\varepsilon>0$ there is a number $1\le\theta(\varepsilon)\le1+\delta$
such that
 % ------------- %
 \begin{equation} \label{2nd_est}
\hspace{-5.5em} \int_1^\infty y^{p/(p+2)}\psi_{j,p}^2\left(\frac{\theta(\varepsilon)}{y^{p/(p+2)}},y\right)
\D y <\frac{1}{\delta}\int_1^\infty\int_{y^{-p/(p+2)}}^{(1+\delta)y^{-p/(p+2)}}
x^py^p\psi_{j,p}^2(x,y)\,\D x\,\D y+\varepsilon\,. %\qquad(3.3)
 \end{equation}
 % ------------- %
Changing variables in the integral on the right-hand side we get
 % ------------- %
\begin{eqnarray*}
\lefteqn{ \int_1^\infty\int_{y^{-p/(p+2)}}^{(1+\delta)y^{-p/(p+2)}}
x^py^p\psi_{j,p}^2(x,y)\,\D x\,\D y} \\ && =\int_1^\infty \int_1^{1+\delta}\frac{y^p\,z^p}{y^{p^2/(p+2)} \cdot y^{p/(p+2)}}\, \psi_{j,p}^2\left(\frac{z}{y^{p/(p+2)}},y\right)\, \D z\,\D y
\\ &&
=\int_1^\infty\int_1^{1+\delta}y^{p/(p+2)}\,z^p\, \psi_{j,p}^2\left(\frac{z}{y^{p/(p+2)}},
y\right)\,\D z\,\D y \\ && \ge\int_1^\infty\int_1^{1+\delta}y^{p/(p+2)}\psi_{j,p}^2\left(\frac{z}
{y^{p/(p+2)}},y\right)\,\D z\,\D y
\\ &&
=\int_1^{1+\delta}\int_1^\infty y^{p/(p+2)}\psi_{j,p}^2\left(\frac{z}{y^{p/(p+2)}},y\right)
\,\D y\,\D z  \\ && \ge\delta\inf_{1\le z\le1+\delta}\left\{\int_1^\infty y^{p/(p+2)}\psi_{j,p}^2
\left(\frac{z}{y^{p/(p+2)}},y\right)\,\D y\right\},
\end{eqnarray*}
 % ------------- %
which proves the validity of inequality (\ref{2nd_est}). Let us proceed to the proof of the inequality (\ref{1st_est}). We fix
a positive $\varepsilon$ and the corresponding number $\theta(\varepsilon)$. In view of Newton-Leibnitz theorem and Cauchy's inequality we have
 % ------------- %
\begin{eqnarray*}
\hspace{-6em}\lefteqn{\int_1^\infty\int_0^{(1+\delta)y^{-p/(p+2)}}
y^{2p/(p+2)}\psi_{j,p}^2(x,y)\,\D x\,\D y=
\int_1^\infty\int_0^{\theta(\varepsilon)y^{-p/(p+2)}}y^{2p/(p+2)}\psi_{j,p}^2(x,y)
\,\D x\,\D y}
\\ &&
+\int_1^\infty\int_{\theta(\varepsilon)y^{-p/(p+2)}}^{(1+\delta)y^{-p/(p+2)}}
y^{2p/(p+2)}\psi_{j,p}^2(x,y)\,\D x\,\D y
\\ && \hspace{-6.5em}
=\int_1^\infty\!\int_0^{\theta(\varepsilon)y^{-p/(p+2)}} \hspace{-1em} y^{2p/(p+2)}
\left(-\int_x^{\theta(\varepsilon)y^{-p/(p+2)}}
\!\frac{\partial\psi_{j,p}}{\partial t}(t,y)\,\D t+\psi_{j,p}\left(\theta(\varepsilon)
y^{-p/(p+2)},y\right)\right)^2\D x\,\D y
\\ && \hspace{-6.5em}
+\int_1^\infty\!\int_{\theta(\varepsilon)y^{-p/(p+2)}}^{(1+\delta)y^{-p/(p+2)}}y^{2p/(p+2)}
\left(\int_{\theta(\varepsilon)y^{-p/(p+2)}}^x
\!\frac{\partial\psi_{j,p}}{\partial t}(t,y)\,\D t+\psi_{j,p}\left(\theta(\varepsilon)
y^{-p/(p+2)},y\right)\right)^2\D x\,\D y
\\ && \hspace{-6.5em}
\le2\int_1^\infty\int_0^{\theta(\varepsilon)y^{-p/(p+2)}}
y^{2p/(p+2)}\left(\int_x^{\theta(\varepsilon)
y^{-p/(p+2)}}\frac{\partial\psi_{j,p}}{\partial t}(t,y)\,\D t\right)^2\,\D x\,\D y
\\ && \hspace{-5em}
+2\int_1^\infty\int_0^{\theta(\varepsilon)y^{-p/(p+2)}}y^{2p/(p+2)}
\psi_{j,p}^2\left(\theta(\varepsilon)y^{-p/(p+2)},y\right)\,\D x\,\D y
\\ && \hspace{-5em}
+2\int_1^\infty\int_{\theta(\varepsilon)y^{-p/(p+2)}}^{(1+\delta)y^{-p/(p+2)}}
y^{2p/(p+2)}\left(\int_{\theta(\varepsilon)
y^{-p/(p+2)}}^x\frac{\partial\psi_{j,p}}{\partial t}(t,y)\,\D t\right)^2\,\D x\,\D y
\\ && \hspace{-5em}
+2\int_1^\infty\int_{\theta(\varepsilon)y^{-p/(p+2)}}^{(1+\delta)y^{-p/(p+2)}}
y^{2p/(p+2)}\psi_{j,p}^2\left(\theta(\varepsilon)y^{-p/(p+2)},y\right)\,\D x\,\D y
\\ && \hspace{-6.5em}
\le2\theta(\varepsilon)\int_1^\infty\int_0^{\theta(\varepsilon)y^{-p/(p+2)}}
y^{p/(p+2)}\int_x^{\theta(\varepsilon)y^{-p/(p+2)}}\left(\frac{\partial\psi_{j,p}}
{\partial t}\right)^2(t,y)\,\D t\,\D x\,\D y
\\ && \hspace{-5em}
+2\theta(\varepsilon)\int_1^\infty y^{p/(p+2)}\psi_{j,p}^2\left(\theta(\varepsilon)
y^{-p/(p+2)},y\right)\,\D y
\\ && \hspace{-5em}
+2(1+\delta-\theta(\varepsilon))\int_1^\infty\int_{\theta(\varepsilon)y^{-p/(p+2)}}^{(1+\delta)
y^{-p/(p+2)}}y^{p/(p+2)}\int_{\theta(\varepsilon)y^{-p/(p+2)}}^x
\left(\frac{\partial\psi_{j,p}}{\partial t}\right)^2(t,y)\,\D t\,\D x\,\D y
\\ && \hspace{-5em}
+2(1+\delta-\theta(\varepsilon))\int_1^\infty y^{p/(p+2)}\psi_{j,p}^2\left(\theta(\varepsilon)
y^{-p/(p+2)},y\right)\,\D y
\\ && \hspace{-6.5em}
\le2\left(\theta(\varepsilon)\right)^2\int_1^\infty\int_0^{\theta(\varepsilon)y^{-p/(p+2)}}
\left(\frac{\partial\psi_{j,p}}{\partial t}\right)^2(t,y)\,\D t\,\D y
\\ && \hspace{-5em}
+2\theta(\varepsilon)\int_1^\infty y^{\frac{p}{p+2}}\psi_{j,p}^2\left(\theta(\varepsilon)
y^{-p/(p+2)},y\right)\,\D y
\\ && \hspace{-5em}
+2\left(1+\delta-\theta(\varepsilon)\right)^2\int_1^\infty
\int_{\theta(\varepsilon)y^{-p/(p+2)}}^{(1+\delta)y^{-p/(p+2)}}
\left(\frac{\partial\psi_{j,p}}{\partial t}\right)^2(t,y)\,\D t\,\D y
\\ && \hspace{-5em}
+2(1+\delta-\theta(\varepsilon))\int_1^\infty y^{p/(p+2)}\psi_{j,p}^2\left(\theta(\varepsilon)
y^{-p/(p+2)},y\right)\,\D y
\\ && \hspace{-6.5em}
\le\frac{5}{2}(1\!+\!\delta)^2 \int_1^\infty\!\int_0^\infty\left(\frac{\partial\psi_{j,p}}{\partial t}
\right)^2(t,y)\,\D t\,\D y+2(1\!+\!\delta)\int_1^\infty y^{p/(p+2)}\psi_{j,p}^2\left(\theta(\varepsilon)
y^{-p/(p+2)},y\right)\,\D y\,.
\end{eqnarray*}
 % ------------- %
Furthermore, by virtue of inequality (\ref{2nd_est}) we infer that
 % ------------- %
\begin{eqnarray*}
\hspace{-6em} \lefteqn{\int_1^\infty\int_0^{(1+\delta)y^{-p/(p+2)}} \hspace{-1em} y^{2p/(p+2)} \psi_{j,p}^2(x,y)\,\D x\,\D y\le\frac{5}{2}(1+\delta)^2\int_1^\infty\int_0^\infty
\left(\frac{\partial\psi_{j,p}}{\partial x}\right)^2(x,y)\,\D x\,\D y}
\\[.5em] &&
+2\frac{1+\delta}{\delta}\int_1^\infty\int_0^{(1+\delta)y^{-p/(p+2)}}
x^py^p\psi_{j,p}^2(x,y)\,\D x\,\D y+2(1+\delta)\varepsilon\,,
\end{eqnarray*}
 % ------------- %
from which, using the arbitrariness of $\varepsilon$, the validity of inequality (\ref{1st_est}) follows. In the same way one proves
 % ------------- %
 \begin{eqnarray*}
\lefteqn{\hspace{-6em}\int_1^\infty\int_0^{(1+\delta)x^{-p/(p+2)}} \hspace{-1em} x^{2p/(p+2)}\psi_{j,p}^2(x,y)\,\D y\,\D x \le
\frac{5}{2}(1+\delta)^2\int_1^\infty\int_0^\infty\left(\frac{\partial\psi_{j,p}}{\partial y}\right)^2\!\!(x,y)
\,\D y\,\D x } \\[.5em] &&
+2\frac{1+\delta}{\delta}\int_1^\infty\int_0^{(1+\delta) x^{-p/(p+2)}}x^py^p\psi_{j,p}^2(x,y)\,\D y\,\D x
 \end{eqnarray*}
 % ------------- %
and analogous bounds for the other hyperbolic ``valleys'' where $x$ or $y$ take negative values. Using the fact that $\|\psi_{j,p}\|=1$ in combination with simple estimates we get
 % ------------- %
 \begin{eqnarray*}
\lefteqn{\int_{\mathbb{R}^2}(x^2+y^2)^{\frac{p}{p+2}}\psi_{j,p}^2(x,y)\,\D x\,\D y} \\ &&
\le\Biggl(\int_{|y|\ge1}\int_{|x|\le(1+\delta)|y|^{-p/(p+2)}}
|y|^{2p/(p+2)}\psi_{j,p}^2(x,y)\,\D x\,\D y
 \\ &&
\qquad +\int_{|y|\ge1}
\int_{|x|>(1+\delta)|y|^{-p/(p+2)}}|y|^{2p/(p+2)}\psi_{j,p}^2(x,y)\,\D x\,\D y
 \\ &&
\qquad +\int_{|x|\ge1}\int_{|y|\le(1+\delta)|x|^{-p/(p+2)}}
|x|^{2p/(p+2)}\psi_{j,p}^2(x,y)\,\D y\,\D x
 \\ &&
\qquad +\int_{|x|\ge1}\int_{(1+\delta)|x|^{-p/(p+2)<|y|<1}}
|x|^{2p/(p+2)}\psi_{j,p}^2(x,y)\,\D y\,\D x\Biggr)+2
 \\ &&
\le(1+\delta)\max\left\{\frac{5}{2}(1+\delta),\frac{2}{\delta}\right\}
\Biggl(\int_{\mathbb{R}^2}\left|\nabla\psi_{j,p}\right|^2(x,y)\,\D x\,\D y
\\ && \qquad  + \int_{\mathbb{R}^2}|xy|^p\psi_{j,p}^2(x,y)\,\D x\,\D y+(1+\delta)^2\Biggr) +2\,, %(3.3)
 \end{eqnarray*}
 % ------------- %
where we have used the fact that $|xy|^p > |y|^{2p/(p+2)}$ holds on the domain of the second one of the four integrals, and a similar bound holds \emph{a fortiori} for the fourth one; the factor $(1+\delta)^2$ prevents from double counting the ``corner regions'' with $|x|,|y|\ge 1$ and $|y|\le(1+\delta)|x|^{-p/(p+2)}$. Choosing $\delta =\frac{-5+\sqrt{105}}{10}$ we get
 % ------------- %
 \begin{eqnarray*}
\lefteqn{\int_{\mathbb{R}^2}(x^2+y^2)^{\frac{p}{p+2}}\psi_{j,p}^2(x,y)\,\D x\,\D y \le
\alpha\Biggl(\int_{\mathbb{R}^2}\left|\nabla\psi_{j,p}\right|^2(x,y)\,\D x\,\D y} \\ && \hspace{9em} +
\int_{\mathbb{R}^2}|xy|^p\psi_{j,p}^2(x,y)\,\D x\,\D y\Biggr)+c\,,
 \end{eqnarray*}
 % ------------- %
where $c:=\alpha(1+\delta)^2+2=2\left(\frac{\alpha^2}{5}+1\right)$.
Since $\lambda_{j,p}$ is the eigenvalue corresponding to the eigenfunction $\psi_{j,p}$ the last relation implies
 % ------------- %
$$
\int_{\R^2}\left|\nabla\psi_{j,p}\right|^2\,\D x\,\D y+\int_{\R^2}|xy|^p\psi_{j,p}^2\,\D x\,\D y\le
\frac{1}{1-\alpha\lambda}(\lambda_{j,p}+c\lambda)\,,\quad j=1,2,\ldots\,. %(3.4)
$$
 % ------------- %
We subtract a number $\varrho$ from both sides of the last equation and express the first integral through the Fourier-Plancherel image of $\psi_{j,p}$. Summing over the first $N$ eigenvalues we obtain
 % ------------- %
 \begin{eqnarray*}
\lefteqn{-\sum_{j=1}^N\int_{R^2}(\varrho-x^2-y^2)|\hat{\psi}_{j,p}|^2\,\D x\,\D y+
\sum_{j=1}^N\int_{R^2}|xy|^p\psi_{j,p}^2\,\D x\,\D y}
\\ &&
\hspace{2.5em} \le\frac{1}{1-\alpha\lambda}\sum_{j=1}^N (\lambda_{j,p}+c\lambda)-N\varrho\,,
 \end{eqnarray*}
 % ------------- %
and the inequality will certainly remain valid if we replace $\varrho-x^2-y^2$ by its positive part $[\varrho-x^2-y^2]_+$. We need the following auxiliary result:

\medskip

\noindent There is a constant $C'_p$ such that for any orthonormal system of real-valued function,
$\Phi=\{\varphi_j\}_{j=1}^N \subset L^2(\R^2),\,N=1,2,\ldots\,$, the inequality
 % ------------- %
$$
\int_{\R^2}\rho_\Phi^{p+1}\,\D x\,\D y\le
C'_p(\ln^pN+1)\sum_{j=1}^N \int_{\R^2}|\xi\eta|^p
|\hat{\varphi_j}|^2\,\D\xi\,\D\eta\,,
$$
 % ------------- %
holds true, where $\rho_\Phi:=\sum_{j=1}^N\varphi_j^2$.

\medskip

\noindent This claim was proved as Theorem~2 in \cite{Bar09} for $p=1$, its extension to any $p\ge 1$ is straightforward. Combining it with  H\"older inequality we obtain the following estimate for the system $\Phi=\{\hat{\psi}_{j,p}\}_{j=1}^N$,
 % ------------- %
 \begin{eqnarray*}
\hspace{-5em} \lefteqn{-\biggl(\int_{\R^2}[\varrho-x^2-y^2]_+^{\frac{p+1}{p}}\,\D x\,\D y\biggr)^{\frac{p}{p+1}}\biggl(\int_{\R^2}\rho_{\hat{\Phi}_p}^{p+1}\,\D x\,\D y\biggr)^{\frac{1}{p+1}}+\frac{C_p'}{1+\ln^p N}
\int_{\R^2}\rho_{\hat{\Phi}_p}^{p+1}\,\D x\,\D y} \\ && \qquad
\le\frac{1}{1-\alpha_p\lambda}\sum_{j=1}^N(\lambda_{j,p}+c\lambda)-N\varrho\,.
 \end{eqnarray*}
 % ------------- %
Consider the function
$f:\: f(z)=-\biggl(\int_{\R^2}{[\varrho-x^2-y^2]_+}^{(p+1)/p}\,
\D x\,\D y\biggr)^{p/(p+1)}z+\frac{C_p'}{1+\ln^p N}z^{p+1}.$
It is easy to verify that its minimum is attained
at the point $z_0\equiv z_0(p)$ where
 % ------------- %
$$
z_0=\frac{{(1+\ln^pN)}^{1/p}}{{(p+1)^{1/p}C_p'}^{1/p}}
\biggl(\int_{\R^2}{[\varrho-x^2-y^2]_+}^{(p+1)/p}
\D x\,\D y\biggr)^{1/p}
$$
 % ------------- %
and substituting the function value at the point $z_0$ we obtain
 % ------------- %
$$
-C_p''{(1+\ln^p N)}^{1/p}\int_{\R^2}{[\varrho-x^2-y^2]_+}
^{(p+1)/p}\,\D x\,\D y \le \frac{1}{1-\alpha_p\lambda}\sum_{j=1}^N(\lambda_{j,p}
+c\lambda)-N\varrho
$$
 % ------------- %
with the constant $C_p''=p (p+1)^{-(p+1)/p}(C_p')^{-1/p}$, which is further equivalent to
 % ------------- %
 \begin{equation} \label{est3}
C_p'''{(1+\ln^pN)}^{1/p}\varrho^{(2p+1)/p} \geq N\varrho-\frac{1}{1-\alpha\lambda}\sum_{j=1}^N(\lambda_{j,p}+c\lambda) %\qquad(3.8)
 \end{equation}
 % ------------- %
with  the new constant given by $C_p''':=C_p''\int_{x^2+y^2\le1}(1-x^2-y^2)^{(p+1)/p}\,\D x\,\D y$.

In the final step we apply Legendre transformation \cite{ME82} to the function at the right-hand side of (\ref{est3}), $g(\varrho):= C_p'''{(1+\ln^pN)}^{1/p} \varrho^{(2p+1)/p}$. By definition we have
 % ------------- %
$$
\widetilde{g}(N)=\sup_{\varrho\geq0}\left(N\varrho-C_p'''\,\varrho^{(2p+1)/p}
(1+\ln^pN)^{1/p}\right)
$$
 % ------------- %
and denoting the expression in the bracket as $h(\varrho)$ we check easily that reaches its maximum at the point  $\varrho_{max}=\left(p/((2p+1)C_p'''\right)^{p/(p+1)} N^{p/(p+1)}
(1+\ln^pN)^{/1/(p+1)}$ and its value there equals
 % ------------- %
$$
\widetilde{g}(N)=h(\varrho_{max})=C_p\frac{N^{(2p+1)/(p+1)}}
{(1+\ln^pN)^{1/(p+1)}}
$$
 % ------------- %
with the constant
 % ------------- %
$$
C_p:=\left(\frac{p}{(2p+1)C_p'''}\right)^{p/(p+1)}\, \frac{p+1}{2p+1}\,.
$$
 % ------------- %
Then (\ref{est3}) implies the bound
 % ------------- %
$$
\frac{1}{1-\alpha_p\lambda}\sum_{j=1}^N(\lambda_{j,p}+c_p\lambda)\ge
h\left(\varrho_{max}\right)=\widetilde{g}(N).
$$
 % ------------- %
which is equivalent to the claim of the theorem. \end{proof}

 % ------------- %
\begin{remark}
{\rm While our main interest concerns sums of eigenvalues, we note that one can use the above result also to derive bounds on more general Lieb-Thirring-type expressions. Indeed, it follows that for large enough natural $N$ the spectrum $\lambda_{1,p}\le\lambda_{2,p}\le\ldots$ of operator $L_p(\lambda)$ satisfies
 % ------------- %
$$
\sum_{j=K}^{K+N}\lambda_{j,p}\geq \frac12\, C_p(1-\alpha\lambda)\,
\frac{N^{(2p+1)/(p+1)}}{(1+\ln^pN)^{1/(p+1)}},\quad K=1,2,\ldots\,.
$$
 % ------------- %
Using this inequality for $K=N$ we infer, in particular, that
 % ------------- %
$$
\lambda_{2N,p}\geq \frac12\, C_p(1-\alpha\lambda)\,
\frac{N^{p/(p+1)}}{(1+\ln^pN)^{1/(p+1)}}\,.
$$
 % ------------- %
Consequently for any positive number $\sigma$ we get
 % ------------- %
$$
\sum_{j=1}^{3N}\lambda_{j,p}^\sigma\geq\sum_{j=2N}^{3N}\lambda_{j,p}^\sigma
\geq N\lambda_{2N,p}^\sigma\geq \frac{C_p^\sigma(1-\alpha\lambda)^\sigma}{2^\sigma}
\frac{N^{(p(\sigma+1)+1)/(p+1)}}{(1+\ln^pN)^{\sigma/(p+1)}}\,,
$$
 % ------------- %
and as a result, the inequality
 % ------------- %
$$
\sum_{j=1}^N\lambda_{j,p}^\sigma\geq\sum_{j=1}^{3[\frac{N}{3}]}\lambda_{j,p}^\sigma
\geq \widetilde{C}_p(1-\alpha\lambda)^\sigma
\frac{N^{(p(\sigma+1)+1)/(p+1)}}{(1+\ln^pN)^{\sigma/(p+1)}}
$$
 % ------------- %
is valid with some positive constant  $\widetilde{C}_p$. }
\end{remark}
 % ------------- %

\subsection{Upper bounds}
\label{ss: upper}

Next we will perform a complementary task and establish an upper bound for spectral sums of $L_p(\lambda),\,p\ge1$, with any subcritical $\lambda$. It will show, in particular, that in the case $\lambda=0$ the asymptotics given by Theorem~\ref{th:lowerbound} is exact up the value of the constant.

 % ------------- %
 \begin{theorem} \label{th:upperbound}
 To any $p\ge 1$ there is a constant $\widetilde C_p$ such that
  % ------------- %
$$
\sum_{j=1}^N\lambda_{j,p}\leq \widetilde C_p\frac{N^{(2p+1)/(p+1)}}
{(1+\ln^pN)^{1/(p+1)}}\,, \qquad N=1,2,\ldots %(4.1)
$$
 % ------------- %
 holds for any $0\le\lambda<\gamma_p$.
 \end{theorem}
 % ------------- %
\begin{proof}[Proof]
Consider the operator $\hat H_p=-\Delta+Q,\quad\textrm{where}\quad Q(x,y)=|xy|^p+|x|^p+|y|^p+1$, in $L^2(\R^2)$. Its spectrum is discrete by Theorem~\ref{thm:subcrit} and minimax principle since $H_p\ge L_p$, and it is obviously sufficient to establish a bound of the above type for the eigenvalues $0\le\beta_{1,p}\le\beta_{2,p}\le\ldots$ of the estimating operator $\hat H_p$.

We shall employ Weyl asymptotics for the number of eigenvalues of below bounded differential operators in a version proved by G.~Rozenblum \cite{Ro74}. Let $T=-\Delta+V$ in $\R^m$, where the potential $V(x)\geq1$ and tends to infinity as $|x|\to\infty$. We denote by $E(\lambda,V)$ the set $\{x\in \R^m:\: V(x)<\lambda\}$ and put
 % ------------- %
$$
\sigma(\lambda,V)=\mathrm{mes\,} E(\lambda,V)\,.
$$\
 % ------------- %
For any unit cube $D\subset \R^m$ we denote the mean value of the function $V$ in $D$ by $V_D$. Furthermore, given a function $f\in L^1(D)$ and $t\le\sqrt{m}$ we define its $L_1$-modulus of continuity by the formula
 % ------------- %
$$
\omega_1(f,t,D):=\sup_{|z|<t}\int_{x\in D,\,x+z\in D}
|f(x+z)-f(x)|\,\D x\,.
$$
 % ------------- %
Then we have the following result \cite{Ro74}:

\medskip

\noindent Suppose that the potential $V$ satisfies the following
conditions:
 % ------------- %
\begin{enumerate}
\item There exists a constant $c$ such that $\sigma(2\lambda,V)\leq c\sigma(\lambda,V)$ holds all $\lambda$ large enough.
 % ------------- %
\item $V(y)\leq cV(x)$ holds if as $|x-y|\leq1$.
 % ------------- %
\item There is a continuous and monotonous function
$\eta:\,t\in[0,\sqrt{m}]\to\R_+$ with $\eta(0)=0$ and a number
$\beta\in[0,\frac12)$ such that for any unit square $D$ we have
 % ------------- %
$$
\omega_1(V,t,D)\leq\eta(t)\,t^{2\beta}\,V_D^{1+\beta}\,.
$$
 % ------------- %
\end{enumerate}
 % ------------- %
Under theses assumption the asymptotic formula
 % ------------- %
 $$
N(\lambda,V)\sim\Gamma_m\Phi(\lambda,V)
 $$
 % ------------- %
holds for the operator $T=-\triangle+V$, where $N(\lambda,V)$ is the number of eigenvalues of $T$ smaller than $\lambda$, unit-ball volume $ \Gamma_m=(2\sqrt{\pi})^{-m} \biggl(\Gamma\biggl( \frac{m}{2}+1\biggr) \biggr)^{-1}$, and
 % ------------- %
 $$
\Phi(\lambda,V)=\int_{\R^2}(\lambda-V)_+^{m/2}\,\D x\,\D y\,,
 $$
 % ------------- %
where as usual $f(x)_+:=\max(f(x),0)$. Let us check that the assumption are satisfied for the operator $\hat H_p$. As for the first one, by definition we have
 % ------------- %
 \begin{eqnarray*}
\hspace{-2em} \lefteqn{\sigma(\lambda,Q) = \int_{|x|^p|y|^p+|x|^p+|y|^p\leq\lambda-1}\,\D x\,\D y =4\int_0^{(\lambda-1)^{1/p}}\int_0^{\frac{(\lambda-x^p-1)^{1/p}}
{(x^p+1)^{1/p}}}\,\D y\,\D x} \\ &&
\hspace{1.2em} =4\int_0^{(\lambda-1)^{1/p}}\frac{(\lambda-x^p-1)^{1/p}}
{(x^p+1)^{1/p}}\,\D x\,, %(4.2)
 \end{eqnarray*}
 % ------------- %
hence
 % ------------- %
 $$
\frac{1}{2p}\lambda^{1/p}\ln\lambda\leq\sigma(\lambda,Q)\leq
\frac{8}{p}\lambda^{1/p}\ln\lambda
 $$
 % ------------- %
holds for large enough $\lambda$ verifying the first assumption. Next we have
 % ------------- %
 $$
|Q(x,y)-Q(x_1,y_1)|\leq12\cdot16^p\,p((x-x_1)^2+(y-y_1)^2)^{1/2}\,Q(x,y) %(4.3)
 $$
 % ------------- %
if $(x-x_1)^2+(y-y_1)^2\leq4$ which gives the second assumption with $c=12\cdot16^p\,p+1$ and the third one with $\beta=0$ and $\eta(t)= 12\cdot16^p\,p\,(12\cdot16^p\,p+1)t$ for $t\in[0,\sqrt{2}]$. Hence the eigenvalues of $\hat H_p$ satisfy
 % ------------- %
$$
N(\lambda,Q)\sim\Gamma_2\Phi(\lambda,Q)\,, %(4.4)
$$
 % ------------- %
where the second factor is equal to
 % ------------- %
 \begin{eqnarray*}
\hspace{-1em} \lefteqn{
\Phi(\lambda,Q)=\int_{\R^2}(\lambda-|xy|^p-|x|^p-|y|^p-1)_+\,\D x\,\D y}
\\ &&
=4\int_0^{(\lambda-1)^{1/p}}\int_0^{\frac{(\lambda-1-x^p)^{1/p}}
{(x^p+1)^{1/p}}}(\lambda-x^py^p-x^p-y^p-1)\,\D y\,\D x
\\ && =4\int_0^{(\lambda-1)^{1/p}}\Biggl((\lambda-1)\frac{(\lambda-x^p-1)
^{1/p}}{(x^p+1)^{1/p}}-\frac{x^p}{p+1}\frac{(\lambda-x^p-1)
^{(p+1)/p}}{(x^p+1)^{(p+1)/p}}
\\ &&
\qquad -x^p\frac{(\lambda-x^p-1)^{1/p}}{(x^p+1)^{1/p}}-
\frac{1}{p+1}\frac{(\lambda-x^p-1)^{(p+1)/p}}{(x^p+1)^{(p+1)/p}}
\Biggr)\,\D x
\\ &&
=\frac{4p}{p+1}\int_0^{(\lambda-1)^{1/p}}\frac{(\lambda-x^p-1)
^{(p+1)/p}}{(x^p+1)^{1/p}}\,\D x\,, %(4.5)
 \end{eqnarray*}
 % ------------- %
which can be estimated as follows
 % ------------- %
$$
\frac{1}{4(p+1)}\lambda^{(p+1)/p}\ln\lambda\leq \Phi(\lambda,Q)
\leq\frac{8}{p+1}\lambda^{(p+1)/p}\ln\lambda %(4.6)
$$
 % ------------- %
implying
 % ------------- %
$$
N(\lambda,Q)\geq \tilde\Gamma_2\,\lambda^{(p+1)/p}\ln\lambda %(4.7)
$$
 % ------------- %
for large enough $\lambda$, where $\tilde\Gamma_2 =\frac{\Gamma_2}{8(p+1)}$. We write $\lambda$ in the form $\lambda= C(\lambda)\biggl(\frac{N(\lambda,Q)}
{\ln N(\lambda,Q)}\biggr)^{p/(p+1)}$ and substitute into the last inequality obtaining
 % ------------- %
 \begin{eqnarray*}
\hspace{-6em} \lefteqn{
N(\lambda,Q)\geq \tilde\Gamma_2 C(\lambda)^{(p+1)/p} \frac{N(\lambda,Q)}
{\ln N(\lambda,Q)} \left( \ln C(\lambda) + \frac{p}{p+1} \ln N(\lambda,Q) - \frac{p}{p+1} \ln \ln N(\lambda,Q) \right)} \\ &&
\hspace{-2.5em} \ge \Gamma_2 C(\lambda)^{(p+1)/p} \frac{N(\lambda,Q)}
{\ln N(\lambda,Q)} \left( \ln C(\lambda) + \frac{p}{2(p+1)} \ln N(\lambda,Q) \right)
 \end{eqnarray*}
 % ------------- %
for $\lambda$ large enough. If $C(\lambda)>1$ we can discard the first term in the last bracket obtaining
 % ------------- %
$$
N(\lambda,Q)\geq \frac{\tilde\Gamma_2}{2(p+1)}\, C(\lambda)^{(p+1)/p}\, N(\lambda,Q)\,,
$$
 % ------------- %
hence
 % ------------- %
$$
C(\lambda)\leq\left(\frac{2(p+1)}{p\tilde\Gamma_2}\right)^{\frac{p}{p+1}}+1\,,
$$
 % ------------- %
which implies
 % ------------- %
$$
\lambda\leq\left( \left(\frac{2(p+1)}{p\tilde\Gamma_2}\right)^{\frac{p}{p+1}} +1 \right)
\left(\frac{N(\lambda,Q)}{\ln N(\lambda,Q)}\right)^{\frac{p}{p+1}} %(4.8)
$$
 % ------------- %
and, \emph{mutatis mutandis}, the following upper bound on the spectrum of the operator $\hat H_p$,
 % ------------- %
$$
\sum_{j=1}^{N(\lambda,Q)}\beta_{j,p}\leq\lambda N(\lambda,Q)
\leq\widetilde{C_p}\frac{(N(\lambda,Q))^{(2p+1)/(p+1)}}
{(\ln N(\lambda,Q)+1)^{p/(p+1)}}
$$
 % ------------- %
with a constant $\widetilde{C_p}$ depending on $p$ only; this yields the claim of the theorem.
\end{proof}

%%%%%%%%%%%%%%%%%%%%%%%%%%%%%%%
%%%% Horn-shaped regions
%%%%%%%%%%%%%%%%%%%%%%%%%%%%%%%
\section{Horn-shaped regions}
\label{s: horn-shape}

Since our bounds are valid for any $p\ge 1$ it is natural to ask about the limit $p\to\infty$ which would correspond to the particle confined in a region with four hyperbolic ``horns'', $D=\{(x,y)\in\R^2:\: |xy|\le1\}$, described by the Schr\"odinger operator $H_D(\lambda):\: H_D(\lambda)\psi =-\Delta\psi-\lambda(x^2+y^2)\psi$ with a nonnegative parameter  $\lambda$ and Dirichlet condition on the boundary $\partial D$. One can prove directly an analogue to Theorem~\ref{th:lowerbound} for this case under a weaker restriction on the coupling constant $\lambda$.

 % ------------- %
 \begin{theorem} \label{th:dirichlet}
The spectrum of $H_D(\lambda)$ is discrete for any $\lambda\in [0,1)$ and the spectral estimate
 % ------------- %
$$
\sum_{j=1}^N\lambda_j\geq\,C(1-\lambda)\frac{N^2}{1+\ln\,N}\,,\qquad N=1,2,\ldots\,, %(5.1)
$$
 % ------------- %
holds true with a positive constant $C$.
 \end{theorem}
 % ------------- %
\begin{proof}[Proof:]
Let us check that for any function $u\in H^1$ satisfying the condition $\left.u\right|_{\partial D}=0$ we have
 % ------------- %
 \begin{equation} \label{direst}
\int_D(x^2+y^2)u^2(x,y)\,\D x\,\D y\le
\int_D\left|\left(\nabla\,u\right)(x,y)\right|^2\,\D x\,\D y\,. %(5.2)
 \end{equation}
 % ------------- %
Using Newton-Leibnitz theorem and Cauchy inequality we get
 % ------------- %
 \begin{eqnarray*}
\hspace{-6em} \lefteqn{
\int_{D_{++}} y^2u^2(x,y)\,\D x\,\D y = \int_0^\infty\int_0^{1/y} y^2u^2(x,y)\,\D x\,\D y}
\\ && \hspace{-5.5em}
=\int_0^\infty\int_0^{1/y}y^2\left(-\int_x^{1/y}
\frac{\partial u}{\partial t}(t,y)\,\D t\right)^2\,\D x\,\D y
\le\int_0^\infty\int_0^{1/y}y\int_x^{1/y}\left(\frac{\partial u}{\partial t} \right)^2(t,y)\,\D t\,\D x\,\D y
\\ && \hspace{-5.5em}
\le\int_0^\infty\int_0^{1/y}y\int_0^{1/y}\left(\frac{\partial u}{\partial t}
\right)^2(t,y)\,\D t\,\D x\,\D y\le\int_0^\infty\int_0^{1/y}\left(\frac{\partial u}{\partial t}
\right)^2(t,y)\,\D t\,\D y\,, %(5.3)
 \end{eqnarray*}
 % ------------- %
where $D_{++} = \{(x,y)\in D:\: x,y\ge 0\}$. Similarly one can prove that
 % ------------- %
$$
\int_{D_{++}}x^2u^2(x,y)\,\D x\,\D y
\le\int_0^\infty\int_0^{1/x}
\left(\frac{\partial u}{\partial t}\right)^2(x,t)\,\D t\,\D x\,; %(5.4)
$$
 % ------------- %
the remaining three cases with $x$ or $y$ taking negative values are treated in a similar way. This proves the inequality (\ref{direst}) which in turn implies
 % ------------- %
$$
H_D(\lambda)\ge-(1-\lambda)\Delta_D\,,  %\quad(5.5)
$$
 % ------------- %
where $\Delta_D$ is the Dirichlet Laplacian on the region $D$. Combining the classical result of B.~Simon \cite{Si83} with the minimax principle we find that the spectrum of $H_D(\lambda)$ is purely discrete for any $\lambda<1$. To prove the sought spectral estimate it is sufficient to check that
 % ------------- %
 \begin{equation} \label{dirbound}
\sum_{j=1}^N\beta_j\ge\,C\frac{N^2}{1+\ln\,N}\,,\qquad N=1,2,\ldots\,,  %\quad(5.6)
 \end{equation}
 % ------------- %
where $\beta_j,\,j=1,2,\ldots\,$, are the eigenvalues of  $\Delta_D$ arranged in the ascending order. The asymptotic eigenvalue distribution of Dirichlet Laplacian $\Delta_D$ for the region in question is well known \cite{JMS92},
 % ------------- %
$$
N(\lambda)\sim\frac{1}{\pi}\lambda\ln\lambda\,,
$$
 % ------------- %
By means of the known inverse asymptotic formula \cite[Sec.~9]{RSS89} we get for the spectrum of $-\Delta_D$ the expression
 % ------------- %
 $$
\beta_j\sim\frac{\pi\,j}{\ln\,j} %\quad(5.7)
 $$
 % ------------- %
as $j\to\infty$, and this in turn yields
 % ------------- %
$$
\sum_{j=1}^N\beta_j
\ge\sum_{j=\left[\frac{N}{2}\right]}^{2\left[\frac{N}{2}\right]}\beta_j
\ge\beta_{\left[\frac{N}{2}\right]}\,\left[\frac{N}{2}\right]
\ge\frac{\pi}{2}\left[\frac{N}{2}\right]^2\frac{1}{\ln\left[\frac{N}{2}\right]}
\ge\frac{\pi}{32}\frac{N^2}{\ln\,N}\,,
$$
 % ------------- %
for all sufficiently large $N$ which proves (\ref{dirbound}) with some constant $C$, and by that the spectral estimate for the operator $H_D(\lambda)$.
\end{proof}

 % ------------- %
\begin{remark}
{\rm The negative term in the Dirichlet case equals $-\lambda r^2$ where $r=(x^2+y^2)^{1/2}$ so the force associated with this potential looks like centrifugal one being proportional to the radius. This brings to mind a tempting picture of a trap which can release particles if it rotates fast enough. A proper description of such a system, however, requires to add to $-\Delta_D$ the operator $\omega\cdot L_\mathrm{z}$ where $\omega$ is the angular velocity, and we will not pursue this idea further in this paper.}
\end{remark}
 % ------------- %

%%%%%%%%%%%%%%%%%%%%%%%%%%
\section*{Acknowledgments}
We are obliged to Milo\v{s} Tater for providing Fig.~\ref{gammap}. The research was supported by the Czech Ministry of Education,
Youth and Sports, and Czech Science Foundation within the projects LC06002 and P203/11/0701.

%%%%%%%%%%%%%%%%%%%%%

\end{document}